\documentclass{llncs}
\usepackage[shortlabels]{enumitem}
\usepackage{comment}
\usepackage{amsmath}
\usepackage{amssymb}
\usepackage{authblk}
\usepackage{hyperref}
\usepackage{color}
\usepackage{url}
\usepackage{float}
\usepackage[inkscape={/Applications/Inkscape.app/Contents/Resources/bin/inkscap??e -z -C}]{svg}

\newcounter{section-preserve}
\newcounter{theorem-preserve}
\newcommand{\blank}[1]{}
\newtoks\magicAppendix
\magicAppendix={}
\newtoks\magictoks
\newif\iflater
\laterfalse
\long\def\later#1{\magictoks={#1}%
  \edef\magictodo{\noexpand\magicAppendix={\the\magicAppendix \par
    \the\magictoks%
  }}
  \magictodo}
\long\def\both#1{\magictoks={#1}%
  \edef\magictodo{\noexpand\magicAppendix={\the\magicAppendix \par
    \noexpand\setcounter{theorem-preserve}{\noexpand\arabic{theorem}}%
    \noexpand\setcounter{theorem}{\arabic{theorem}}%
    \noexpand\setcounter{section-preserve}{\noexpand\arabic{section}}%
    \noexpand\setcounter{section}{\arabic{section}}%
    \noexpand\let\noexpand\oldsection=\noexpand\thesection
    \noexpand\def\noexpand\thesection{\thesection}
    \noexpand\let\noexpand\oldlabel=\noexpand\label
    \noexpand\let\noexpand\label=\noexpand\blank
    \the\magictoks%
    \noexpand\setcounter{theorem}{\noexpand\arabic{theorem-preserve}}%
    \noexpand\setcounter{section}{\noexpand\arabic{section-preserve}}%
    \noexpand\let\noexpand\thesection=\noexpand\oldsection
    \noexpand\let\noexpand\label=\noexpand\oldlabel
  }}
  \magictodo
  \the\magictoks}
\def\magicappendix{\latertrue \the\magicAppendix}

{\makeatletter
 \gdef\xxxmark{%
   \expandafter\ifx\csname @mpargs\endcsname\relax 
     \expandafter\ifx\csname @captype\endcsname\relax 
       \marginpar{xxx}
     \else
       xxx 
     \fi
   \else
     xxx 
   \fi}
 \gdef\xxx{\@ifnextchar[\xxx@lab\xxx@nolab}
 \long\gdef\xxx@lab[#1]#2{{\bf [\xxxmark #2 ---{\sc #1}]}}
 \long\gdef\xxx@nolab#1{{\bf [\xxxmark #1]}}
}

\makeatletter
\newcommand{\removelatexerror}{\let\@latex@error\@gobble}
\makeatother
\begin{document}

\def \isnotin {\nsubseteq}

\def \eps {\varepsilon}

\title{Upward Partitioned Book Embeddings}

\author{Hugo A. Akitaya\inst{1}
	\and Erik D. Demaine\inst{2}
	\and Adam Hesterberg\inst{2} 
	\and \newline Quanquan C. Liu\inst{2}
}

\institute{Tufts University, Medford, MA, USA\\
	\email{hugo.alves\_akitaya@tufts.edu}
	\and
	Massachusetts Institute of Technology,
	Cambridge, MA, USA\\
	\email{{edemaine,achester,quanquan}@mit.edu}
}

\maketitle

\begin{abstract} 
We analyze a directed variation of the book embedding problem when the page partition is prespecified and the nodes on the spine must be in topological order (upward book embedding).
Given a directed acyclic graph and a partition of its edges into $k$
pages, can we linearly order the vertices such that
the drawing is upward (a topological sort) and each page avoids
crossings?  
We prove that the problem is NP-complete for $k\ge 3$, and for $k\ge 4$ even in the special case when each page is a matching.
By contrast, the problem can be solved in linear time for
$k=2$ pages when pages are restricted to matchings.
The problem comes from Jack Edmonds (1997), motivated as a generalization of
the map folding problem from computational origami.
\end{abstract}

\section{Introduction}
\label{sec:intro}
\textbf{Book Embeddings.}
Bernhart and Keinen~\cite{BK79} first introduced the concept of \emph{book embeddings} and \emph{book thickness} of graphs in 1979. Since then, book embeddings and book thickness have been widely studied as natural geometric invariants in directed and undirected graphs with applications in graph drawing and graph algorithms.
Book embeddings (also studied under the name of \emph{stack layouts} \cite{CLR87,HP97,HPT99,HP99}) have applications in VLSI design, fault-tolerant processing, parallel process scheduling, sorting networks, and parallel matrix computations~\cite{CLR87,HLR92,HR92}. 

Given an undirected graph $G=(V, E)$, where $|V| = n$ and $|E| = m$, a \emph{book embedding} consists of
\begin{enumerate}
\item a linear ordering $\pi$ of the vertices~$V$, defining an embedding of the vertices into the \emph{spine} (a line in the plane); and
\item a disjoint partition of the edges $E$ into sets, so that each set of the partition can be embedded into a \emph{page} (half-plane bounded by the spine) without intersection between the edges on each page.
\end{enumerate}
The pages join together at the spine to form a \emph{book}.
The \emph{book thickness} of a graph $G$ is the minimum number $k$ of pages in any book embedding of~$G$.

Much of the previous research on book embedding (on undirected graphs)
focuses on the book thickness of particular graph classes such as complete bipartite graphs~\cite{ENK97} and planar graphs (which have an upper bound of~$4$ pages)~\cite{yannakakis89,BBKR15,BGR16}. 
Graphs of book thickness $1$ turn out to have a simple characterization as
(exactly) outerplanar graphs \cite{BK79},
and such graphs can be recognized in linear time~\cite{W87}.
By contrast, graphs with a $2$-page book embedding are exactly the sub-Hamiltonian
graphs~\cite{BK79}, and recognizing such graphs is NP-complete~\cite{Wig82}.   

\textbf{Directed Graphs.}
Motivated by many of the same applications,
book embedding has been generalized to directed graphs.
For directed acyclic graphs (DAGs), an
\emph{upward book embedding} is a book embedding such that the linear ordering
of the vertices on the spine is in topological order~\cite{HP99,HPT99}.
For general digraphs, \emph{oriented book embeddings}~\cite{M16} require that all arcs embedded into a page (or the spine) must agree in orientation (pointed up or down with respect to the order on the spine).

Research in upward book embedding includes combinatorial results for
classes of DAGs such as trees, cycles, or paths by using characteristics
of the underlying undirected graph~\cite{HP99}. Furthermore, more recent research studied
the book embedding of directed planar graphs~\cite{FFR11}.
As in the undirected case, there is a linear-time algorithm to determine
whether a DAG has a $1$-page upward book embedding~\cite{HP99}
(although the algorithm is very different from the $1$-page book embedding
algorithm applied to the DAG's underlying undirected graph).
Furthermore, determining whether a DAG has a $6$-page upward book embedding
is NP-complete~\cite{HP99}.
There is a linear-time algorithm for $2$-page upward book embedding of
planar directed series-parallel graphs \cite{GDL02}.
Note that undirected series-parallel graphs are necessarily sub-Hamiltonian.
For graphs with cycles, oriented book embeddings and oriented book thickness have been found for several graph classes including cycles and oriented trees~\cite{M16}. 

See~\cite{DW04,DW05} for more detailed citations lists and surveys about book embeddings and linear graph layouts for both undirected and undirected graphs. 

\textbf{Partitioned Problem.}
In the \emph{partitioned} book embedding problem,
we are \emph{given} the partition of edges into pages.
This variation eliminates one of the previous combinatorial aspects of finding a book embedding (namely finding the partition of edges), and leaves only finding the vertex order on the spine. Intuitively, this problem should be simpler.
Indeed, there is a linear-time algorithm for determining whether a given edge partition can result in a $2$-page book embedding of an undirected graph \cite{HN14,ABB12}.
Nonetheless, the partitioned $k$-page book embedding problem is NP-complete for undirected graphs where $k \geq 3$~\cite{ALN15}. 
In their proof, the gadgets only work with undirected edges and therefore cannot directly be applied to upward book embedding. 

In terms of partitioned book embedding problems where the edges in a partition form a matching,~\cite{Hoske12} showed that for undirected graphs, it is NP-hard to find the find the ordering of vertices given unbounded number of pages (i.e. unbounded number of partitions). Although the reduction in~\cite{Hoske12} is also from \textsc{Betweenness} (defined in Def.~\ref{def:betweenness}), the techniques used are simpler and different from ours since they consider undirected graphs and allow unbounded number of partitions of edges. 

\begin{table}[t]
\def\RESULT#1#2{\begin{minipage}{0.9in}\centering\vspace{0.5ex}{#1}\par{\small #2}\vspace{0.5ex}\end{minipage}}
\def\OPEN{\sl OPEN}
\centering
\begin{tabular}{|c|c|c|c|c|}
\hline
Type & $k=1$ & $k=2$ & $k=3$ & $k\geq 4$
\\ \hline
Undirected & \RESULT{$O(n)$}{\cite{BK79}} & \RESULT{$O(n)$}{\cite{HN14}} & \RESULT{NP-complete}{\cite{ALN15}} & \RESULT{NP-complete}{\cite{ALN15}}
\\ \hline
Upward & \RESULT{$O(n)$}{\cite{HP99}} & \OPEN & \RESULT{\bf NP-complete}{[Theorem~\ref{thm:upbe-np-complete}]} & \RESULT{\bf NP-complete}{[Theorem~\ref{thm:upbe-np-complete}]}
\\ \hline
Matching & \RESULT{$O(n)$}{\cite{HP99}} & \RESULT{\boldmath $O(n)$}{[Theorem~\ref{thm:UMPBE-2}]} & \OPEN & \RESULT{\bf NP-complete}{[Theorem~\ref{thm:UMPBE-4}]}
\\ \hline
\end{tabular}
\caption{Summary of known and new results in partitioned book embedding.  New results are written in bold.}
\label{book-embedding-problems}
\end{table}

\textbf{Our Results.}
In this paper, we study the natural combination of the partitioned variation
(where we are given the partition of edges into pages) with upward book
embedding of DAGs, which has not been considered before (\textsc{Upward Partitioned $k$-Page Book Embedding}).
We prove that the resulting 
problem is NP-complete for any $k \geq 3$.
Our hardness proof techniques also apply to a special case of this problem, called
\textsc{Upward Matching-Partitioned $k$-Page Book Embedding},
where only disjoint edges map to each page (forming a matching).
For this special case, we show NP-hardness for $k \geq 4$ and that book embedding can be solved in linear time for $k=1$ page or $k=2$ pages. 
Table~\ref{book-embedding-problems} puts these results in context with
previous results.

\textsc{Upward Matching-Partitioned $4$-Page Book Embedding} is in fact
motivated by the (nonsimple) map folding problem, posed by Jack Edmonds
in 1997 (personal communication with E. Demaine); see
\cite{MapFolding,2xn}.  
Edmonds showed that the problem of finding a flat
folded state of an $m \times n$ grid crease pattern, with specified mountains
and valleys, reduces to exactly this type of book embedding problem,
with the $k=4$ pages corresponding to the four compass directions of a square.
Furthermore, $1 \times n$ and $2 \times n$ map folding reduce to the problems
with $k=2$ and $k=3$ pages.  
Algorithms for solving
\textsc{Upward Matching-Partitioned $k$-Page Book Embedding} are thus
of particular interest because they solve the long-standing map folding open problem
as well.

In Section~\ref{sec:definitions}, we formally define our book embedding models. 
In Section~\ref{sec:main}, we prove NP-completeness for \textsc{Upward Partitioned $3$-Page Book Embedding}. 
Finally, in Section~\ref{sec:other}, we show that \textsc{Upward  Matching-Partitioned Book Embedding} can be solved in linear time for 2 pages and is NP-complete for 4 pages.

\section{Definitions}
\label{sec:definitions}
We define the \textsc{Upward Partitioned $k$-Page Book Embedding} (\textsc{UPBE-$k$}) problem similarly to the definition for \textsc{Partitioned $k$-Page Book Embedding} as given in~\cite{ALN15}. Specifically, we are given a directed acyclic graph (DAG) $G=(V, E)$ and a partition of the edges in~$E$: $P = \left\{E_1, E_2, \dots, E_k\right\}$ where $E_1 \dot{\cup} E_2 \dot{\cup} \cdots \dot{\cup} E_k = E$, where $\dot{\cup}$ denotes disjoint union.  
The goal is to determine whether $G$ can be embedded in a $k$-page book such that the ordering $\pi$ of the vertices on the spine is topologically sorted and each $E_i \in P$ lies in a separate page.

The \textsc{Upward Matching-Partitioned $k$-Page Book Embedding} (\textsc{UMPBE-$k$}) problem is the special case of \textsc{UPBE-$k$} in which every edge partition $E_i \in P$ forms a directed matching, that it, has at most one edge incident to each vertex in~$G$.

For a given upward partitioned book embedding instance $G=(V, E, P)$, let $\pi$ represent a valid ordering of $V$ on the spine where a valid ordering is one that satisfies the constraints on every page (e.g. non-crossing edges) and follows topological order. 
As stated previously, $\pi$ is also a valid topological sorting of~$V$. 
We write $\pi(x) < \pi(y)$ (resp., $\pi(x) > \pi(y)$) if node $x\in V$ comes earlier/before (resp., later/after) $y\in V$ in $\pi$. 
For ease of wording, we will assign colors to edge partitions and refer to edges within each partition to have a particular color.

\section{UPBE is NP-Complete}
\label{sec:main}
We show that \textsc{UPBE-$k$} is NP-hard via a reduction from the NP-complete problem \textsc{Betweenness}. 
The problem \textsc{Betweenness} is defined as follows. 

\begin{definition}[\textsc{Betweenness}~\cite{O79}]\label{def:betweenness}
We are given a set $L$ of $n$ elements and a set $C$ of $m$ ordered triples where each member of a triple is a member of $L$. 
Let $\phi$ be a total ordering of the elements in $L$.
An ordered triple, $\langle a, b, c \rangle \in C$ is satisfied if either $\phi(a) < \phi(b) < \phi(c)$ or $\phi(c) < \phi(b) < \phi(a)$ is true. 
The goal is to find an ordering $\phi$ such that all ordered triples in $C$ are satisfied. 
\end{definition}

Given an instance $(L, C)$ of \textsc{Betweenness}, we construct an instance, $G=(V, E)$, with edge partition, $P = \left\{\texttt{Red}, \texttt{Green}, \texttt{Blue}\right\}$, of \textsc{UPBE-$3$} such that a subsequence of a solution $\pi$ corresponds to a valid ordering $\phi$ of $L$ in the \textsc{Betweenness} instance.
For each element $x\in L$, we create vertices $x_1,\ldots, x_{2m - 1}$ in $V$. 
We call these vertices the \emph{element vertices}. 
Our reduction uses two types of gadgets: \emph{ordered triple gadgets} and \emph{order preserving gadgets}.
Their function is to enforce a betweenness constraint given by an element of $C$ and to ensure that the order of element vertices of subscript $j$ in $\pi$, with $j\in\{1,\ldots, 2m - 2\}$, is the reverse order of element vertices of subscript $j+1$, respectively.
We prove that $(G,P)$ admits an upward book embedding given by $\pi$ if and only if the order of element vertices of the same subscript in $\pi$  corresponds to a solution $\phi$ for the $(L, C)$ \textsc{Betweenness} instance. 

\subsection{Gadgets}
\label{sec:gadgets}
\noindent\textbf{Ordered Triple Gadget.}
We order the set $C$ arbitrarily.
For the $(\frac{i+1}{2})$-th ordered triple $\langle a, b, c \rangle\in C$, where $i$ is an odd integer between and including $1$ and $2m-1$, we construct an \emph{ordered triple gadget} that enforces the betweenness constraint on the triple of element vertices $a_{i}, b_{i}, c_{i}$ in $\pi$.
Specifically, we create the following nodes and edges.

\begin{definition}[Ordered Triple Gadget]\label{def:ordered-triple}
Let $(L, C)$  be an instance of \textsc{Betweenness}. Order the set $C$ arbitrarily. For the $\left(\frac{i+1}{2}\right)$-th ordered triple $\langle a, b, c\rangle$, where $i$ is an odd integer between and including $1$ and $2m-1$, construct nodes $l_i$, $\alpha_i$, $\omega_i$, $a'_i$, $b'_i$, $c'_i$, and $h_i$. Then, create directed edges $(l_i, \alpha_i), (l_i, \omega_i) \in \mathtt{Red}$, $(\alpha_i, a_i'), (\alpha, b_i'), (\omega_i, b_i'), (\omega_i, c_i') \in \mathtt{Blue}$, and $(a_i', h_i), (b_i', h_i), (c_i', h_i) \in \mathtt{Green}$. 
\end{definition}

Refer to Fig.~\ref{fig:UPBE-3} for an example construction.
Nodes $a_{i}', b_{i}'$ and $c_{i}'$ are respectively connected to $a_{i}'', b_{i}''$ and $c_{i}''$ by an edge in \texttt{Red} (where $a_i''$, $b_i''$, and $c_i''$ are part of the order preserving gadget defined in Def.~\ref{def:order-preserving}).
The topologically earliest and latest nodes in the gadget are respectively $l_i$ and $h_i$.
The choice between $\pi(a'_{i})<\pi(b'_{i})<\pi(c'_{i})$ and $\pi(a'_{i})>\pi(b'_{i})>\pi(c'_{i})$ (and hence the choice between $\pi(a_i) < \pi(b_i) < \pi(c_i)$ and $\pi(a_i) > \pi(b_i) > \pi(c_i)$) is encoded in the choice between 
$\pi(\alpha_{i})<\pi(\omega_{i})$ and $\pi(\alpha_{i})>\pi(\omega_{i})$ as we prove in Lemmas~\ref{lem:ordered-triple} and~\ref{lem:element-vert}.

\begin{figure}[h]
	\centering
	\def\svgwidth{.7\columnwidth}
	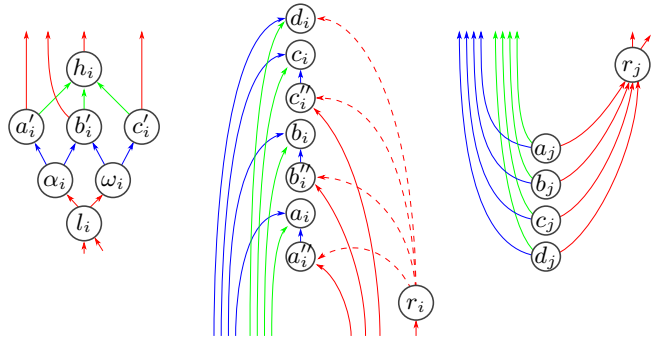
	\caption{Ordered triple gadget (left) and order preserving gadget for odd $i$ (center) and for even $j$ (right). The red edges from $a_i'$, $b_i'$ and $c_i'$ from the ordered triple gadget lead to $a_i''$, $b_i''$, and $c_i''$ in the order preserving gadget. The red edge from $h_i$ directs to $r_i$.}
	\label{fig:UPBE-3} 
\end{figure}

\begin{figure}[h]
\centering
\def\svgwidth{.5\columnwidth}
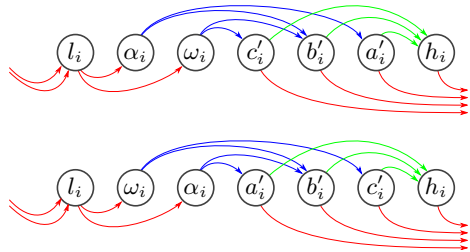
\caption{The two possible embeddings of the ordered triple gadget.}
\label{fig:UPBE-3-order} 
\end{figure}

\begin{lemma}\label{lem:ordered-triple}
Given a positive instance $(G,P)$ containing the ordered triple gadget shown in Fig.~\ref{fig:UPBE-3} (left), if $\pi(h_i)< \min(\pi(a_i''), \pi(b_i''), \pi(c_i''))$ then either $\pi(a_i'') < \pi(b_i'') < \pi(c_i'')$ or $\pi(c_i'') < \pi(b_i'') < \pi(a_i'')$.
\end{lemma}

\begin{proof}
Notice that $h_i$ must appear after $a_i'$, $b_i'$ and $c_i'$ due to topological order.
By the assumption in the lemma, $h_i$ be before $a_i''$, $b_i''$, $c_i''$ respectively. 
We first prove that $\min(\pi(a_i'), \pi(b_i'), \pi(c_i')) > \max(\pi(\alpha_i), \pi(\omega_i))$, i.e., no vertex in the gadget must occur between $\alpha_i$ and $\omega_i$. 
Because of the topological order, $\pi(b_i') > \max(\pi(\alpha_i), \pi(\omega_i))$. 
Suppose $\pi(\alpha_i)<\pi(\omega_i)$ (see Figure~\ref{fig:UPBE-3-order} (top)).
By topological order $\pi(c_i') > \pi(\omega_i)=\max(\pi(\alpha_i), \pi(\omega_i))$.
If $\pi(a_i') < \pi(\omega_i)$, since $\pi(\omega_i)<\pi(a_i'')$, 
the red edges $(a_i', a_i'')$ and $(l_i, \omega_i)$ would intersect, a contradiction.
The case when $\pi(\alpha_i)>\pi(\omega_i)$ is symmetric.

Trivially, either $\pi(\alpha_i) < \pi(\omega_i)$ or $\pi(\alpha_i) > \pi(\omega_i)$. 
We first assume $\pi(\alpha_i) < \pi(\omega_i)$ (Fig.~\ref{fig:UPBE-3-order} (top)). 
Then $\pi(c_i') < \pi(b_i') < \pi(a_i')$, or else at least one pair of blue edges ($(\alpha_i, b_i')$, $(\alpha_i,  a_i')$, $(\omega_i, c_i')$, $(\omega_i, b_i')$) would cross.
Since all red edges $(a_i',a_i''), (b_i',b_i'')$ and $(c_i',c_i'')$ nest around $h_i$, every pair of such edges must be nested.
Therefore, $\pi(a_i'') < \pi(b_i'') < \pi(c_i'')$.
The case when $\pi(\alpha_i)>\pi(\omega_i)$ is symmetric. \qed
\end{proof}

\noindent\textbf{Order Preserving Gadget.}
By Lemma~\ref{lem:ordered-triple}, each ordered triple gadget enforces a betweenness constraint on vertices $a_i''$, $b_i''$ and $c_i''$.
The \emph{order preserving gadgets} serve two purposes: ensuring that the $i$-th betweenness constraint is enforced in the $i$-th copy of element vertices; and ensuring that each copy of element vertices must occur in the reverse order of its predecessor.
That implies that every other copy of element vertices occur in exactly the same order.
We build $2m - 1$ order preserving gadgets, the $j$-th gadget containing $x_j$ for each $x\in L$.

\begin{definition}[Order Preserving Gadget]\label{def:order-preserving}
For each odd $i$ in the range $[1, 2m-1]$, we build the following nodes and edges:
\begin{enumerate}
\item Nodes $a_i'', b_i'', c_i''$, $x^p_i$ for $p \in [1, n]$, and $r_i$.
\item Directed edges $(a_i'', a_i), (b_i'', b_i), (c_i'', c_i) \in \mathtt{Red}
$. 
\item Directed paths of length $7$ connecting $r_i$ to $x^p_i$ for all $p \in [1, n]$ and where $x^p_i \neq a_i, b_i, c_i$. The paths alternate between red and green edges. 
\item Directed paths of length $7$ connecting $r_i$ to $a''_i$, $b''_i$, and $c''_i$ with alternating red and green edges. 
\end{enumerate}

For each even $j$ in the range $[1, 2m-1]$, we build the following nodes and edges:
\begin{enumerate}
\item Nodes $x^p_j$ for $p \in [1, n]$ and $r_j$.
\item Directed edges $(x^p_j, r_j) \in \mathtt{Red}$ for all $p \in [1, n]$.
\end{enumerate}
\end{definition}

The gadget is divided into two parts: the \emph{elements part} containing the element vertices, and the \emph{order preserving tree} whose root is labeled $r_i$ or $r_j$.
Fig.~\ref{fig:UPBE-3} (center) shows an example of an order preserving gadget containing element vertices with odd subscript.
Such instances are connected to ordered triple gadgets by three incoming red edges and have the vertex $r_i$ as the lowest vertex in the topological order.
The dashed edges represent a path of length 7 of alternating red/green edges that are connected to the element vertex $x_i$ if $x\in L$ is not in the $\left(\frac{i+1}{2}\right)$-th ordered triple, or connected to the vertex $x_i''$ otherwise.
The vertices $x_i''$ for an element $x$ in the $i$-th ordered triple are then connected to the element vertex $x_i$ by a blue edge.
For gadgets that contain element vertices with even subscript $j$ (Fig.~\ref{fig:UPBE-3} (right)), $r_j$ is the highest vertex in the topological order.
For even $j\in\{2,\ldots,2m-2\}$, we connect $x_j$ to $x_{j-1}$ with a blue edge and $x_j$ to $x_{j+1}$ with a green edge, for all $x\in L$ (see Fig.~\ref{fig:example-full-construction}).

For odd $i$, the order preserving tree consists of $n$ paths of length 7 of alternating red/green edges connected to the $i$-th element vertex (represented in Fig.~\ref{fig:UPBE-3} (center) as dashed arrows).
Informally, their purpose is to allow such paths to ``cross" the vertices connected to the $i$-th ordered triple gadget by red edges, while $r_i$ as the first vertex in the topological order of the order preserving gadget.

\begin{lemma}\label{lem:element-vert}
	Let $(G,P)$ be a positive instance containing an order preserving gadget of odd index $i$ connected to an ordered triple gadget representing $\langle a,b,c\rangle$. 
	If there exists a set of blue edges $(s_{i-1},x_i)$, $\pi(s_{i-1})<\pi(r_i)$ and $\pi(s_{i-1}) < \min(\pi(a_i'', b_i'', c_i''))$, from some vertices $s_{i-1}$ for all $x\in L$, then either $\pi(a_i)<\pi(b_i)<\pi(c_i)$ or $\pi(a_i)>\pi(b_i)>\pi(c_i)$. 
\end{lemma}

\begin{proof}
	By topological order, $\pi(r_i)<\pi(x_i)$ for all $x\in L$.
	Then, all blue edges of the form $(s_{i-1},x_i)$ must nest around $r_i$.
	By Lemma~\ref{lem:ordered-triple}, the order of vertices $a_i''$, $b_i''$ and $c_i''$ must obey the betweenness constraint $\langle a,b,c\rangle$.
	Since $\pi(y_i'')>\pi(r_i), y\in\{a,b,c\}$, if $\pi(a_i'')<\pi(b_i'')$ then $\pi(a_i)<\pi(b_i)$ or else edges $(a_i'', a_i)$ and $(s_{i-1}, b_i)$ would cross.
	With similar arguments, we can show that the order of the $i$-th element vertices must obey the betweenness constraint $\langle a,b,c\rangle$.\qed
\end{proof}

\begin{lemma}\label{lem:preserving-gadget}
Let $(G,P)$ be a positive instance containing three subsequent order preserving gadgets with indices $j-1$, $j$, and $j+1$ where $j$ is an even integer in $\left\{2, \dots, 2m - 2\right\}$. 
If $\pi(r_{j}) < \min\{\pi(r_{j-1}), \pi(r_{j+1})\}$, the element vertices with subscript $j+1$ appear in $\pi$ in the same order as the element vertices with subscript $j-1$. 
\end{lemma}

\begin{proof}
We prove that the order of the element vertices with subscript $j-1$ is the reverse order of the element vertices with subscript $j$ in $\pi$.
By a similar argument, we can then show the same for $j$ and $j+1$, completing the proof.
Notice that, since $j$ is even,  $\pi(x_{j}) < \pi(r_{j})$, $\pi(r_j) < \pi(r_{j-1})$, and $\pi(r_{j-1})<\pi(x_{j-1})$ for all $x\in L$ due to the topological order of the vertices.
By the assumption in the lemma, all blue edges $(x_j,x_{j-1})$ nest around $r_j$ and $r_{j-1}$.
Therefore, any pair of such edges must be nested or they would cross. 
Hence, if $\pi(b_{j}) > \pi(a_{j})$, then $\pi(b_{j-1})<\pi(a_{j-1})$ or the blue edges $(b_j, b_{j-1})$ and $(a_j, a_{j-1})$ would cross (see Fig.~\ref{fig:example-full-construction} for example). Therefore, the order of the $(j-1)$-th copy of element vertices is the reverse order of the $j$-th copy. The same argument holds for $x_j$ and $x_{j+1}$ for all $x \in L$. Given that the order of the $(j-1)$-th and the order of the $(j+1)$-th
copy of the element vertices are in reverse order of the $j$-th copy, the order of the $(j-1)$-th and $(j+1)$-th copies of the element vertices must be the same. \qed
\end{proof}

The next corollary immediately follows from Lemma~\ref{lem:preserving-gadget}. 

\begin{corollary}\label{cor:preserving-same}
If $\pi(r_{j}) < \min\{\pi(r_{j-1}), \pi(r_{j+1})\}$ for all even $j \in \{2, \dots, 2m-2\}$, then all element vertices with even subscript, $j \in \{2, \dots, 2m-2\}$ appear in the same order and all element vertices with odd subscript, $i \in \{1, \dots, 2m-1\}$ appear in the same order. Furthermore, all element vertices with even subscript, $j$, appear in the reverse order of all element vertices with odd subscript, $i$.
\end{corollary}

\subsection{Final Reduction}\label{sec:final-construct}

We create $n$ ordered triple gadgets as defined in Def.~\ref{def:ordered-triple} and $2m-1$ order preserving gadgets as defined in Def.~\ref{def:order-preserving}, connecting via the following set of edges:

\begin{enumerate}
\item $(a'_i, a''_i), (b'_i, b''_i), (c_i, c''_i) \in \mathtt{Red}$ for all odd $i \in [1, 2m-1]$. 
\item $(r_j, l_{j-1}), (r_j, l_{j+1}) \in \mathtt{Red}$ for all even $j \in [1, 2m-1]$.
\item $(s, x^p_n) \in \mathtt{Blue}$ for all $p \in [1, n]$.
\item $(x^p_j, x^p_{j-1}) \in \mathtt{Blue}$ for all $p \in [1, n]$ and even $j \in [1, 2m-1]$. 
\item $(x^p_j, x^p_{j+1}) \in \mathtt{Green}$ for all $p \in [1, n]$ and even $j \in [1, 2m-1]$. 
\item $(h_i, r_i) \in \mathtt{Red}$ for all odd $i \in [1, 2m-1]$. 
\item $(x^p_j, r_j) \in \mathtt{Red}$ for all $p \in [1, n]$ and even $j \in [1, 2m-1]$.
\item $(r_{2m-2}, s), (s, l_{2m-1}) \in \mathtt{Blue}$.
\item $(r_{2m-2}, l_{2m-1}) \in \mathtt{Red}$.
\end{enumerate}

We connect the ordered triple and order preserving gadgets as described above, obtaining an instance $(G,P)$ of \textsc{UPBE-$k$}. Using this reduction, we prove that UPBE-$k$ is NP-complete for $k \geq 3$.

\both{
\begin{theorem}\label{thm:upbe-np-complete}
\textsc{UPBE-$k$} is NP-complete for $k \geq 3$.
\end{theorem}
}

The proof follows from Lemmas~\ref{lem:ordered-triple}, \ref{lem:element-vert}, and~\ref{lem:preserving-gadget} and the constructions of the gadgets; please refer to the Appendix for the full proof. 

\later{
\begin{proof}
\textsc{UPBE-$k$} is in NP since an ordering $\pi$ has $O(n)$ size. Whether it is a valid ordering can be checked in $O(n)$ time.
We now show NP-hardness.
Given an instance of \textsc{Betweenness}, $(L, C)$, the constructed graph $(G, P)$ as defined in our reduction in Section~\ref{sec:final-construct} is a positive instance of \textsc{UPBE-$k$} if and only if $(L, C)$ is a positive instance of \textsc{Betweenness}. 	
The topological order of $G$ ensures that the conditions in Lemmas \ref{lem:ordered-triple}, \ref{lem:element-vert}, \ref{lem:preserving-gadget} and Corollary~\ref{cor:preserving-same} are always met.
By these lemmas, if $(G,P)$ admits a valid ordering $\pi$, the ordering of a copy of element vertices corresponds to a valid ordering $\phi$ of the \textsc{Betweenness} instance $(L,C)$. 

We now show that if $(L,C)$ admits a valid ordering $\phi$, we can obtain a valid ordering $\pi$ for $(G,P)$.
Refer to Figure~\ref{fig:example-full-construction}.
We order all element vertices of even indices using the same order of $\phi$.
For even $j$, set $\min\{\pi(x_j)\}=\frac{j-2}{2}(n+1)$ and $\pi(r_j)=\frac{j}{2}(n+1)-1$. 
Set $\pi(s)=(m-1)(n+1)+1$.
We now embed the vertices with odd indices $i$ starting from $2m-1$ to $1$.
We embed all the vertices of the $i$-th ordered triple gadget before the corresponding order preserving gadget.
All element vertices should occur in the reverse order of $\phi$.
Embed the ordered triple gadget as in Figure~\ref{fig:UPBE-3-order} so that the order of $a_j', b_j'$ and $c_j'$ is the same as $a,b$ and $c$ in $\phi$.
Embed the order preserving gadget as shown in Figure~\ref{fig:embedding-order-gadget} where a dashed arrow to $x_i''$ represents a path whose vertices appear subsequently in $\pi$ right before $x_i''$.
The paths from $r_i$ must be realized so that they are parallel appearing from bottom to top (using the orientation of Figure~\ref{fig:embedding-order-gadget}) in the same order that its corresponding element vertex appear in $\pi$.
Hence, no paths pairwise cross.
A red and a green edge is used to ``go through" each $x_i''$ that occur before the path's endpoint.
Such an embedding allows the alternating red/green paths to be realized so that no crossing occur with the three red edges coming from the ordered triple gadget.
By construction, $\pi$ is set so that no two edges of the same color between element vertices and vertices in ordered triple gadgets of the same index cross.
Furthermore, the blue (resp. green, resp. red) edges connecting element vertices of different indices will be all nested and will not pairwise cross (see Figure~\ref{fig:example-full-construction}) since all ordered triple gadgets are satisfied by the ordering of the odd order preserving gadgets (which are all ordered the same and are the reverse order of all even order preserving gadgets). \qed
\end{proof}

\begin{figure}[h!]
	\centering
	\def\svgwidth{.7\columnwidth}
	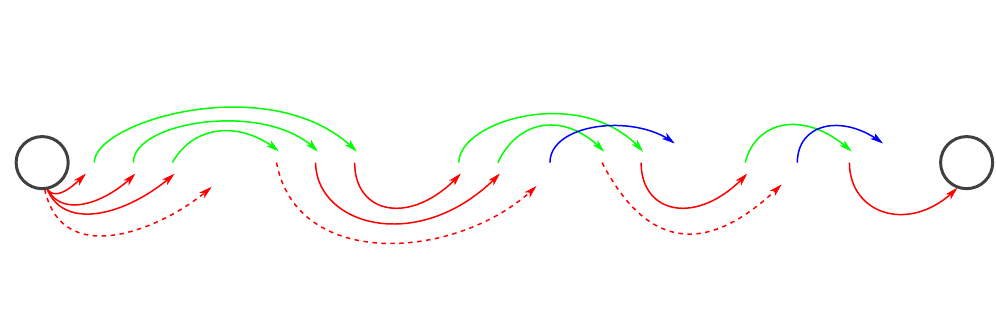
	\caption{An embedding of the order preserving gadget shown in Fig.~\ref{fig:UPBE-3} (center).}
	\label{fig:embedding-order-gadget}
\end{figure}
}

\section{\textsc{UMPBE}}\label{sec:other}

In this section, we discuss the \textsc{Upward  Matching-Partitioned Book Embedding} problem, where given an instance $(G,P)$, each set of the partition $P$ induces a subgraph in $G$ that is a matching (i.e. no vertex is incident to more than one edge of each set of the partition).
We first show \textsc{UMPBE}-$4$ is NP-complete and then show that \textsc{UMPBE}-$2$ is solved in $O(n)$ time.
When $|P|=1$, the algorithm in~\cite{HP99} for \textsc{UPBE}-$1$ can also solve \textsc{UMPBE}-$1$ in $O(n)$ time.

\subsection{\textsc{UMPBE}-$4$}

\begin{theorem}
\label{thm:UMPBE-4}
	\textsc{UMPBE}-$k$ is NP-complete for $k\ge 4$.
\end{theorem}
\begin{proof}

As with \textsc{UPBE}, it is clear that \textsc{UMPBE} is in NP since an order $\pi$ of vertices of $G$ serves as a certificate.
We show NP-hardness by reducing from \textsc{Betweenness}, adapting the proof of Theorem~\ref{thm:upbe-np-complete} to \textsc{UMPBE}-$4$.
We again refer to the partitions in $P=\{\texttt{Red}, \texttt{Blue},\texttt{Green},\texttt{Yellow}\}$ as colors.
The gadgets adapted from Section~\ref{sec:gadgets} are shown in Figure~\ref{fig:UMPBE-gadget}.

\begin{figure}[h!]
	\centering
	\def\svgwidth{\columnwidth}
	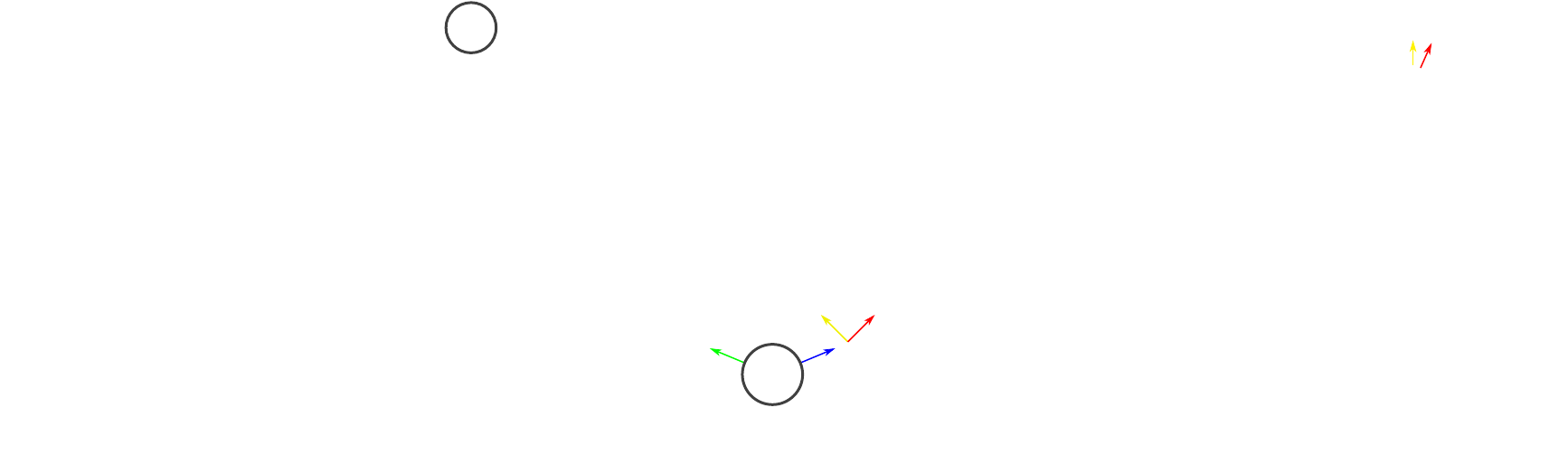
	\caption{Gadgets for the reduction to \textsc{UMPBE}-$4$.}
	\label{fig:UMPBE-gadget}
\end{figure}

For odd $i$ in $\{1,\ldots,2m-1\}$ we connect gadgets with yellow edges $(h_i, r_i)$ and $(r_{i-1},l_i)$ (if $i>1$), and with the red edge $(r_{i+1},l_i)$ (if $i<2m-1$).
Lemmas~\ref{lem:ordered-triple} holds for the new gadget replacing $x_i''$ by $x_i$.
We omit its proof due to the similarity.
The dashed arrows in Figure~\ref{fig:UMPBE-gadget} represent paths of alternating colors as described in the next paragraph.
Lemma~\ref{lem:preserving-gadget} also trivially holds.
Therefore, given a valid order $\pi$ of vertices of $G$, the order of element vertices corresponds to a solution $\phi$ of the \textsc{Betweenness} instance.

It remains to show that, given a solution $\phi$ of the \textsc{Betweenness} instance, we can obtain a solution $\pi$ for the produced instance.
The order in which the gadgets are embedded are the same as in the proof of Theorem~\ref{thm:upbe-np-complete} and, therefore, no edge between gadgets cross.
We now show that each gadget has a cross-free embedding using $\phi$.
The embedding of the ordered triple gadget is very similar to that shown in Figure~\ref{fig:UPBE-3-order} and we chose $\pi(\alpha_i)>\pi(\omega_{i})$ or $\pi(\alpha_i)<\pi(\omega_{i})$ depending on whether $a$ appears before $c$ or vice-versa in $\phi$.
In the order preserving gadget, we use the same order (resp., reverse order) of $\phi$ for even (resp., odd) $j$.

\begin{figure}[h!]
	\centering
	\def\svgwidth{.7\columnwidth}
	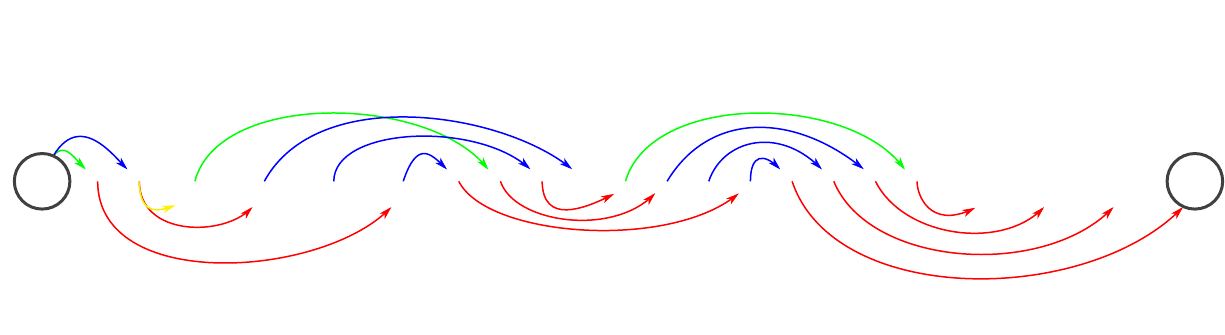
	\caption{Embedding the order preserving gadget in Figure~\ref{fig:UMPBE-gadget} (center).}
	\label{fig:UMPBE-embedding}
\end{figure}

Notice that the order preserving gadget now contains a binary tree and we cannot chose arbitrarily the order of vertices $x_j^*$ for $x\in L$. We call this tree the \emph{binary order preversing tree} that ensures that $\pi(x_i) > \pi(r_i)$ and $\pi(x_j) < \pi(r_j)$ for all $x \in L$ for odd $i$ and even $j$.
The tree is obtained by ordering the vertices $x_j^*$ arbitrarily and building a binary tree that alternates between green/blue and yellow/red edges in order for the induced subgraph of each color to be a matching.
The paths connecting $x_j^*$ to $x_j$ allows us to order the element vertices $x_j$ using $\phi$, independent of the order of vertices $x_j^*$.
We construct such paths in the following way. Each of the paths contains $n$ edges. 
These paths alternate between green/blue and yellow/red edges starting with the opposite color group from the last row of the tree.
Let $x_j^*$ be the $t$-th vertex, $t\in\{1,\ldots,n-1\}$,  in an arbitrary order chosen as the order of the leaves of the order preserving tree.
The colors of edges alternate between blue and red except for the $t$-th edge of the $t$-th vertex, such that the edge is green (resp., yellow) if the alternation would make it blue (resp., red).
The embedding of the paths can be obtained by changing the order of the paths in an insertion sort manner, considering the last path (the path from $d_j^*$ to $d_j$ in Figure~\ref{fig:UMPBE-embedding}) as the first element in the array and adding paths one by one in increasing order (see Fig.~\ref{fig:UMPBE-embedding}). Let $x_j^{t}$ be the $t$-th vertex of the path from $x_j^*$ to $x_j$ and let $y_j^t$ be the $t$-th vertex, $t\in\{1,\ldots,n-1\}$,  in the order chosen as the order of the leaves of the order preserving tree.
Assume that the set $A$ of  all $x_j^{t-1}$ is embedded so that the vertices are contiguous in the spine.
We embed the set $B$ of vertices of the form $x_j^{t}$ immediately after the vertices in $A$ in the reverse order in which they appear in $A$, apart from $y_j^{t}$.
This order guarantees that no crossing is induced since all edges of the same color $(x_j^{t-1}, x_j^{t})$ are nested in parallel from $A$ except for $(y_j^{t-1}, y_j^{t})$ which is of a different color.
We can thus add $y_j^{t}$ in any place in the ordering of $B$. \qed
\end{proof}

\subsection{\textsc{UMPBE}-$2$}

\begin{theorem}
\label{thm:UMPBE-2}
	Given an instance $(G,\{E_1,E_2\})$ where $G=(V,E)$, $E=E_1 \dot{\cup} E_2$ and both $(V,E_1)$ and $(V,E_2)$ are matchings, \textsc{UMPBE} can be solved in $O(n)$ time where $n=|V|$.
\end{theorem}

\begin{proof}
In positive instances, $G = (V, E)$ must be 2-page book embeddable and therefore planar~\cite{BK79}.
Hence $|E|=O(n)$.
Every connected component of the undirected version of $G$ must be either a path or a cycle, or else the induced subgraph of the partitions would not be a matching. 
Furthermore, the edges connected in such paths or cycles must alternate in color.
Each connected component can be solved separately since the concatenation of the solution (total order on vertices) of connected components is a solution of the original problem.
Without loss of generality, we consider only the case when $G$ is connected.

We provide a reduction to 1D origami when $G$ is a path and a reduction to single vertex flat foldability if $G$ is a cycle. The reduction runs in $O(n)$ time producing an instance with $n-1$ creases.
Both the flat foldability of 1D origami and single vertex flat foldability can be determined in linear time~\cite{MapFolding,Bern&Hayes}. A \emph{face} in an 1D origami is defined as a segment in the 1D origami and a \emph{crease} is defined as a place where the origami can be folded. A \emph{face} in a single vertex crease pattern is the space between creases. A 1D origami is defined by a line while a single vertex crease pattern is defined by a single vertex where rays originating from the vertex represent creases.

For both the case of the path and the cycle, we create an instance of 1D origami and single vertex flat foldability, respectively, in the following way. For each edge $e\in E$ we create a mountain crease if $e \in \mathtt{Red}$ and a valley crease if $e \in \mathtt{Blue}$. Each face of the produced instance represents a vertex in $G$. The reduction will thus produce an instance where each face of the origami has the same length, which can be viewed as a linkage formed by identical bars.
If $G$ is a path (resp., cycle), the output will be a list (resp., circular list) containing the assignment of the creases on a line segment (resp., a single vertex origami).
Start with one endpoint of the path or with an arbitrary vertex of the cycle.
Traverse the undirected version of $G$ using BFS.
For each edge traversed add \emph{mountain} (resp., \emph{valley}) to the end of the list if the traversed edge corresponds to an edge in $E_1$ with the same direction of the traversal or to an edge in $E_2$ in the opposite direction (resp., corresponds to an edge in $E_2$ with the same direction of the traversal or to an edge in $E_1$ in the opposite direction).
Since every edge is traversed once, the size of the list is $n-1$ (resp., $n$ if a cycle). Thus, the only difference between single vertex crease patterns and 1D origami is that the faces form a cycle as opposed to a line segment, respectively. 

Due to the similarity of the reduction models for paths and cycles, it suffices to show the equivalence between instances when $G$ is a path.
As previously stated, each face of the paper corresponds to a vertex in $G$. Each crease represents an edge in $G$ and whether the crease is a mountain fold or a valley fold in the final state of the origami determines the partition of the edges of $G$ into \texttt{Red} or \texttt{Blue} edges. 
If we consider the starting vertex as the leftmost face of the unfolded paper and that $f_1$ is not flipped in the folded state, a mountain fold puts the adjacent face $f_2$ below $f_1$. 
Without loss of generality, the edge in $G$ that represents the connection between $f_1$ and $f_2$ is in $E_2$ and points from $f_2$ to $f_1$.
By repeating the argument for every edge, we conclude that $G$ represents the above/below relation of faces of the folded state of the 1D origami and $E_1$ (resp., $E_2$) represents the creases that lie right (resp., left) of the folded state (see Figure~\ref{fig:1D-origami}).
Then, it is easy to verify that the origami is flat-foldable iff $(G,\{E_1,E_2\})$ is a positive instance of \textsc{UMPBE}. \qed
\end{proof}

\begin{figure}[h!]
	\centering
	\def\svgwidth{.5\columnwidth}
	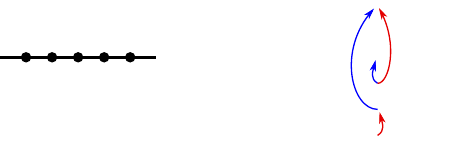
	\caption{A 1D origami crease pattern is shown (left) with mountain/valley labeled as M/V respectively, together with its folded state (center) and the corresponding \textsc{UMPBE}-$2$ instance (right).}
	\label{fig:1D-origami}
\end{figure}

\begin{figure}[h!]
	\centering
	\def\svgwidth{\columnwidth}
	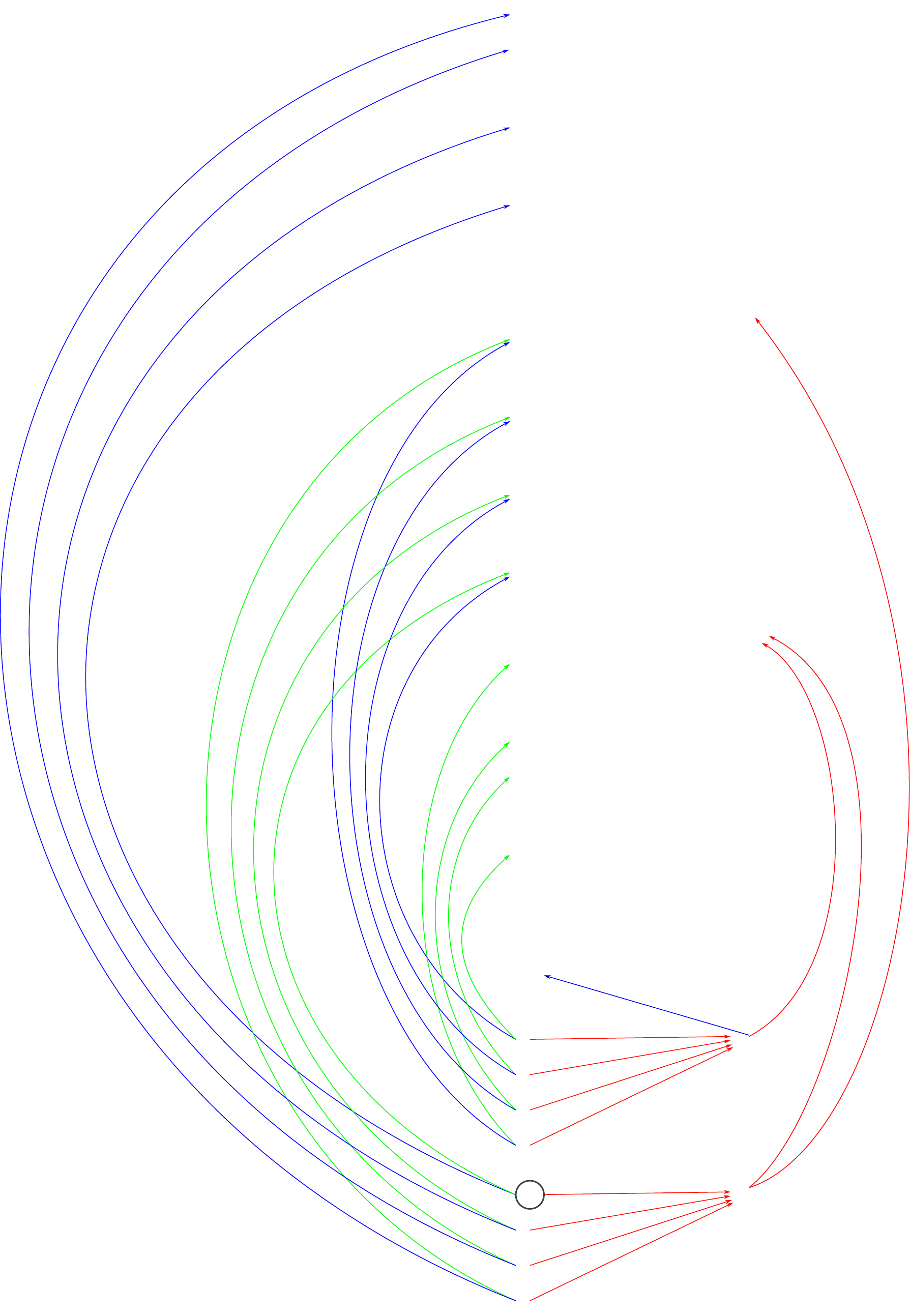
  \caption{Example of a full construction of a reduction. Here, the instance of 					\textsc{Betweenness} is $(L, C)$ where $L=\left\{a, b, c, d\right\}$ and $C=\left\{\langle a, b, c\rangle, \langle b, c, d\rangle, \langle d, b, a\rangle\right\}$.}
	\label{fig:example-full-construction}
\end{figure}

\clearpage
\section*{Acknowledgements}

We thank Jack Edmonds for valuable discussions in August 1997 where he
described how \textsc{Upward Matching-Partitioned $k$-Page Book Embedding}
generalizes the map folding problem.
We also thank Therese Biedl for valuable discussions in 2007
about the complexity this problem.

This research was conducted during the 31st Bellairs Winter Workshop on Computational Geometry which took place in Holetown, Barbados on March 18--25, 2016.
We thank the other participants of the workshop for helpful discussion and for providing a fun and stimulating environment. We also thank our anonymous referees for helpful suggestions in improving the clarity of our paper.

Supported in part by the NSF award CCF-1422311 and Science without Borders. Quanquan Liu is supported in part by NSF GRFP under Grant No. (1122374).

\bibliographystyle{alpha}
\bibliography{ref}
\section*{Appendix}
\magicappendix

\end{document}